\documentclass{ics}

\usepackage{amsmath,amsthm,amsfonts,amssymb}
\usepackage{graphicx}

\def\eps{\epsilon}

\def\th{\theta}

\def\Om{\Omega}

\def\to{\rightarrow}

\def\calm{{\cal M}}

\newcommand{\reals}{\mbox{$\mathbb{R}$}}

\newcommand{\alg}{{\cal A}}
\newcommand{\mech}{{\cal M}}

\newcommand{\balloc}{X}
\newcommand{\ballocs}{{\mathbf \balloc}}

\newcommand{\balloci}[1][i]{{\balloc_{#1}}}

\newcommand{\val}{v}
\newcommand{\vals}{{\mathbf \val}}
\newcommand{\valsmi}{{\mathbf \val}_{-i}}
\newcommand{\vali}[1][i]{{\val_{#1}}}

\newcommand{\util}{u}

\newcommand{\utili}[1][i]{{\util_{#1}}}

\newcommand{\type}{t}
\newcommand{\types}{{\mathbf \type}}

\newcommand{\typei}[1][i]{{\type_{#1}}}

\newcommand{\crit}{\theta}

\newcommand{\criti}[1][i]{{\crit_{#1}}}

\newcommand{\decl}{d}
\newcommand{\decls}{{\mathbf \decl}}
\newcommand{\declsmi}{\decls_{-i}}
\newcommand{\decli}[1][i]{{\decl_{#1}}}

\newtheorem{thm}{Theorem}[section]

\newtheorem{lem}[thm]{Lemma}
\newtheorem{prop}[thm]{Proposition}
\newtheorem{claim}[thm]{Claim}
\newtheorem{conj}[thm]{Conjecture}

\theoremstyle{definition}

\theoremstyle{remark}

\newtheorem{example}[thm]{Example}

\DeclareMathOperator*{\argmax}{arg\,max}

\begin{document}
\title{Beyond Equilibria: Mechanisms for Repeated Combinatorial Auctions}
\author{%
  Brendan Lucier$^{1}$ 
}
\address{%
  $^{1}$Dept of Computer Science, University of Toronto}
\email{%
  blucier@cs.toronto.edu
  }

\begin{abstract} 
We study the design of mechanisms in combinatorial auction domains.  We focus on settings where the auction is repeated, motivated by auctions for licenses or advertising space.  We consider models of agent behaviour in which they either apply common learning techniques to minimize the regret of their bidding strategies, or apply short-sighted best-response strategies.  We ask: when can a black-box approximation algorithm for the base auction problem be converted into a mechanism that approximately preserves the original algorithm's approximation factor on average over many iterations?  We present a general reduction for a broad class of algorithms when agents minimize external regret.  We also present a new mechanism for the combinatorial auction problem that attains an $O(\sqrt{m})$ approximation on average when agents apply best-response dynamics.
\end{abstract}


\maketitle





\section{Introduction}

We consider problems in the combinatorial auction (CA) domain, where $m$ objects are to be allocated among $n$ potential buyers in order to maximize total value, subject to problem-specific feasibility constraints.  
These packing problems are complicated by game-theoretic issues: the buyers 
might benefit from misrepresenting their values to an allocation algorithm.  This prompts us to design mechanisms that use payments to encourage reasonable behaviour.
%
The well-known VCG mechanism solves incentive issues by inducing truth-telling as a dominant strategy, but is infeasible for 
computationally intractible problems.
Indeed, for many interesting problems (such as combinatorial auctions), there are large gaps between the best-known approximation factors attainable by efficient truthful mechanisms and those possible in purely computational settings.  
For some problems, these large gaps are essential \cite{PSS-08}.

In this paper we consider the problem of designing mechanisms that implement approximations to combinatorial auction problems without the use of dominant-strategy truthfulness.  We are motivated by the domain of repeated auctions, where an auction problem is resolved multiple times with the same objects and bidders. These include, for example, auctions for advertising spaces or slots \cite{eos-05}, bandwidth auctions (such as the FCC spectrum auction), airline landing rights auctions \cite{CSS-05}, etc.  In these settings a mechanism for the (one-shot) auction problem corresponds to a repeated game to be played by the agents.\footnote{
A simple extension allows preferences and participants to change over time, but sufficiently slowly compared to the rate of auction repetition.  Our results should extend easily to such settings.}

The question of how to model agent behaviour in repeated games has been studied extensively in the economic and algorithmic game theory literature (see chapters 17-21 of \cite{AGT} and references therein).  Many proposed models suppose that agents choose strategies (or distributions thereover) at equilibrium, where no agent has incentive to unilaterally deviate.  There are, however, a number of reasons to believe such models are unrealistic: in general equilibria are computationally hard to find, and may not exist without the presence of agents who randomize over strategies for no reason other than to preserve the stability of the system.  Even when pure equilibria exist, agents may not necessarily converge to an equilibrium (of the single-round game) or agree on which equilibrium (of the extended-form game) to choose.  Such concerns have also been noted elsewhere \cite{BHLR-08, GMV-05}.

We will instead focus our attention on two models of agent behaviour that do not make equilibrium assumptions, and have gained recent interest in the algorithmic game theory literature.  In the first model, agents can play arbitrary sequences of strategies for the repeated auction, under the assumption that they obtain low regret relative to the best fixed strategy in hindsight.  More precisely, the average external regret of such a bidder must tend to $0$ as the number of auction rounds increases.  These regret-minimizing bidders can be seen as agents that learn how to bid intelligently (relative to any fixed strategy benchmark) from the bidding history of past auction iterations.  Note that we require no assumptions about the synchrony or asynchrony of updates; arbitrary sets of agents can update their strategies concurrently.  The regret-minimization assumption is realistic because simple, efficient algorithms exist that minimize external regret for linear optimization problems such as repeated auctions \cite{KV-05,KKL-07}.  Under this model, our goal is to design an auction mechanism that achieves an approximation to the optimal social welfare \emph{on average} over sufficiently many rounds of the repeated auction.  This is precisely the problem of designing a mechanism with bounded \emph{price of total anarchy}, as introduced by Blum et al \cite{BHLR-08}.

In the second model, we assume that agents choose myopic best-response strategies to the current strategies of the other agents.  Such bidding behaviour is best motivated in settings where agents update their declarations asynchronously (ie. not concurrently).  We therefore model this behaviour as follows: on each auction round, an agent is chosen uniformly at random, and that agent is given the opportunity to change his strategy to the current myopic best-response.  Under this model, our goal again is to design auction mechanisms that achieve approximations to the best possible social welfare \emph{on average} over sufficiently many auction rounds, with high probability over the random choices of bidders.  This is closely related to the concept of the \emph{price of (myopic) sinking}, as introduced by Goemans et al \cite{GMV-05}.

Our high-level goal is to decouple computational issues from incentives issues.  A full (and admittedly ambitious) solution in our domain would be a black-box conversion of a given approximation algorithm into a mechanism that implements\footnote{Throughout the paper we use ``implement'' in the economic sense of obtaining the desired properties when used by rational agents.} the same approximation ratio, on average over sufficiently many auction rounds, given our model of bidder behaviour.  Our primary research question, partially addressed herein, is to what extent such implementations are possible.

\subsection{Our Contribution}

We design mechanisms that are based on a particular class of approximation algorithms for combinatorial auction problems: those that are monotone and satisfy the loser-independence property.  An algorithm is \emph{monotone} if, whenever a bidder can win some set $S$ by declaring a value of $v$ for it, then he could also win any subset of $S$ with any declared value at least $v$.  This monotonicity condition characterizes truthfulness when bidders are single-minded (meaning that each agent has value for only a single set), but not for general auction problems \cite{LOS-99}.
Roughly speaking, an algorithm is \emph{loser-independent} if 
the outcome for an agent depends only on those agents who would win if he did not participate, and on their declared values for their winnings.
This extends a notion of loser-independence for single-parameter problems, introduced by Chekuri and Gamzu \cite{CG-09}, to general auction problems.  Many interesting algorithms satisfy these properties, including greedy algorithms for CAs \cite{LOS-99} and convex bundle auctions \cite{BB-04}, primal-dual algorithms for unsplittable flow \cite{BKV-05}, and others.

Our first main result is that any monotone loser-independent $c$-approximate algorithm can be implemented as a mechanism with price of total anarchy at most $c(1+o(1))$.  
Our mechanism is a black-box reduction from an algorithm for a one-shot auction iteration, and the same mechanism is applied each auction round.  The form of our mechanism is very simple: on each round, it applies a simple modification to the bidders' declarations, then runs the approximation algorithm on the modified declarations and charges critical prices (i.e. an agent who wins a set pays the smallest amount he could have declared for that set and won it, given the declarations of the other bidders).

Our implementation does not depend on the specific algorithms used by the agents to minimize their regret; only that their regret vanishes as the number of rounds increases.  The rate of convergence to our approximation bound will depend on the rate at which the agents' regret vanishes.  

We demonstrate that our mechanism is resilient to the presence of byzantine agents, in the following sense.  If each agent either applies regret-minimizing strategies or makes arbitrary declarations (but never declares more than his true value for a set), then the mechanism attains a $c(1+o(1))$ approximation to the optimal welfare \emph{obtainable by the regret-minimizing bidders}.  The no-overbidding assumption is necessary (as otherwise a byzantine agent could bid arbitrarily highly and prevent any welfare from being obtained) and realistic, since we can view byzantine players as not understanding how to participate intelligently in the auction and thus likely to bid conservatively.

We conjecture that the mechanism described above also implements an $O(c)$ approximation, on average over sufficiently many rounds, in the model of best-response bidders.  
Whether this is so is an important open question.  

We then focus specifically on the general combinatorial auction problem in the best-response model.  
Specifically, we present a mechanism that implements an $O(s)$ approximation for combinatorial auctions with set allocations of size at most $s$, then extend this to a mechanism that implements an $O(\sqrt{m})$ approximation for general combinatorial auctions. Note that this approximation factor is the best possible, in that we attain it with high probability after polynomially many auction rounds (in fact, only a slightly superlinear number of rounds). 
%
%
We point out that while \emph{truthful} mechanisms with similar approximation ratios are known for \emph{single-minded} combinatorial auctions, our results are significant improvements over what is known to be achieveable with deterministic truthful algorithms.


Our results require a mild game-theoretic assumption, which is that bidders will not apply strategies that are (strictly) dominated by easily-found alternatives.  This is precisely the assumption of algorithmically undominated strategies, as introduced by Babaioff et al \cite{BLP-09}.  Additionally, the mechanism for best-response bidders also applies a technique known in implementation theory as \emph{virtual implementation}, where an alternative social choice rule is applied with vanishingly small probability \cite{JAC-01}.  We view this not as an introduction of randomness into the algorithm being implemented, but rather as the introduction of a trembling-hand consideration into the solution concept that encourages reasonable behaviour when best-response agents must distinguish between otherwise equally beneficial strategies.

\subsection{Regret Minimization}

We now describe the concept of external regret minimization in further detail.  The external regret of a sequence of declarations is the difference between the average utility of an agent (i.e. value of goods received minus payment extracted) and the maximum average utility that could have been obtained by a single fixed declaration made each round.  An algorithm for generating declarations is regret-minimizing if its regret vanishes as a function of the number of auction rounds.  

A simple and efficient algorithm due to Kalai and Vempala \cite{KV-05} minimizes regret for linear optimization problems, even when the strategy space has exponential size.  The algorithm requires access to an exact best-response oracle.  Kakade et al \cite{KKL-07} show how to use a $\gamma$-approximate best response oracle to achieve a $\gamma$-approximation to the best fixed declaration in hindsight.  

The regret-minimization mechanism we construct in this paper will reduce the strategic choices of bidders to a simple linear optimization problem, so that the algorithms described above can be used by agents to minimize external regret. This requires access to best-response oracles for the approximation algorithms being implemented.  The ability to compute best-response choices is also a necessary component for the setting of best-response bidders, where the best-response assumption certainly requires that such strategies can be found efficiently.  It is easy to compute best-responses when agents choose between only polynomially many strategies (such as, for example, when each agent's true valuation is a combination of polynomially many desired bundles).  In general, however, the problem of computing best-response is non-trivial when agents have exponentially many strategic choices; we leave the construction of such oracles to the creators of particular auction mechanisms.  

\subsection{Related Work}

Truthful mechanisms for the combinatorial auction problem have been extensively studied.  For general CAs,
Hastad's well-known inapproximability result \cite{HAS-97} 
shows that it is hard to approximate the problem to within $\Om(\sqrt{m})$ assuming $NP \neq ZPP$.  
The best known deterministic truthful mechanism for CAs with general valuations 
attains an approximation ratio of $O(\frac{m}{\sqrt{\log m}})$ \cite{HKNS-04}.  A randomized $O(\sqrt{m})$-approximate mechanism that is truthful in expectation was given by Lavi and Swamy \cite{LS-05}.  Dobzinski, Nisan and Schapira \cite{DNS-06} then gave an $O(\sqrt{m})$-approximate universally truthful randomized mechanism.

Many variations on the combinatorial auction problem have been considered in the literature.  Bartal et. al. \cite{BGN-03} give a truthful $O(m^{\frac{1}{B-2}})$ mechanism for multi-unit combinatorial auctions with $B$ copies of each object, for all $B \geq 3$.  Dobzinski and Nisan \cite{DN-07} construct a truthful $2$-approximate mechanism for multi-unit auctions. 
Many other problems have truthful mechansisms  
(\cite{LOS-99, MN-08, BKV-05}) 
when bidders are restricted to being single-minded.  In \cite{BL-09} the authors study the limited power of certain classes of greedy algorithms for truthfully approximating CA problems.

The problem of designing combinatorial auction mechanisms that implement approximations at equilibria (and, in particular, Bayes-Nash equilibria for partial information settings) was considered in \cite{CKS-08} for submodular CAs, and in \cite{BL-10} for general CA problems.  Implementation at equilibrium, especially for the alternative goal of profit maximization, has a rich history in the economics literature; see, for example, Jackson \cite{JAC-01} for a survey.

The study of regret-minimization goes back to the work of Hannan on repeated two-player games \cite{H-57}.  Kalai and Vempala \cite{KV-05} extend the work of Hannan to online optimization problems, and Kakade et al \cite{KKL-07} further extend to settings of approximate regret minimization.  Blum et al \cite{BHLR-08} apply regret-minimization to the study of inefficiency in repeated games, coining the phrase ``price of total anarchy'' for the worst-case ratio between the optimal objective value and the average objective value when agents minimize regret.

Properties of best-response dynamics in repeated games, and especially the question of convergence to a pure equilibrium, is well-studied (see Chapter 19 of \cite{AGT}).  The study of average performance of best-response dynamics as a metric of game inefficiency, the so-called ``price of sinking,'' was introduced by Goemanns et al \cite{GMV-05}.

Babaioff et al \cite{BLP-09} study implementation of algorithms in undominated strategies, which is a relaxation of the dominant strategy truthfulness concept.  They focus on a variant of the CA problem in which agents are assumed to have ``single-value'' valuations, and present a mechanism to implement such auctions in a multi-round fashion.  By comparison, mechanisms in our proposed model solve each instance of an auction in a one-shot manner, and our solution concept assumes that the auction is repeated multiple times.

\section{Model and Definitions}


In general we will use boldface to represent vectors, subscript $i$ to denote the $i$th component, and subscript $-i$ to denote all components except $i$, so that, for example, $\vals = (\vali,\valsmi)$.  

We consider the domain of combinatorial auction problems, where $n$ agents desires subsets of a set $M$ of $m$ objects.  An \emph{allocation profile} is a collection of subsets $X_1, \dotsc, X_n$, where $X_i$ is thought of as the subset allocated to agent $i$.  A particular problem instance is defined by the set of \emph{feasible allocation profiles} that are permitted; for example, the general combinatorial auction problem requires that all allocated subsets be disjoint.  Each agent $i$ has a privately-held \emph{valuation function} $\typei : 2^M \to \reals$, his \emph{type}, that assigns a value to each allocation.  We assume that valuation functions are monotone and normalized so that $\val(\emptyset) = 0$.  
A valuation function $\val$ is \emph{single-minded} if there exists $S \subseteq M$ and $x \geq 0$ such that $\val(T) = x$ if $S \subseteq T$ and $\val(T) = 0$ otherwise.  We will write $\emptyset$ for the zero valuation, and $(S,x)$ for a single-minded declaration for $S$ at value $x$.

An \emph{allocation rule} $\alg$ assigns to each valuation profile $\vals$ a feasible outcome $\alg(\vals)$; we write $\alg_i(\vals)$ 
for the allocation to agent $i$.  
We write $\alg$ for both an allocation rule and an algorithm that implements it.

An allocation rule is \emph{loser-independent} if, whenever $\valsmi$, $\valsmi'$ satisfy $\alg(\emptyset,\valsmi) = \alg(\emptyset,\valsmi')$ and  $\val_j(\alg_j(\emptyset,\valsmi)) = \val_j(\alg_j(\emptyset,\valsmi'))$ for all $j \neq i$, then $\alg(\vali,\valsmi) = \alg(\vali,\valsmi')$.  In other words, agent $i$'s perception of the behaviour of $\alg$ depends only on those agents who would win if agent $i$ did not participate, and on their declared values for their winnings.
%

A payment rule $P$ assigns a vector of $n$ payments to each valuation profile.  A \emph{direct revelation mechanism} $\mech$ is composed of an allocation rule $\alg$ and a payment rule $P$.  The mechanism proceeds by eliciting a valuation profile $\decls$ of \emph{declarations} from the agents, then applying the allocation and payment rules to $\decls$.  
The utility of agent $i$ for mechanism $\calm$, given declaration profile $\decls$, is $\utili(\decls) = \typei(\alg_i(\decls)) - P_i(\decls)$.  We think of each agent as wanting to choose $\decli$ to maximize $\utili(\decls)$.

The \emph{social welfare} obtained by allocation profile $\ballocs$, given type profile $\types$, is $SW(\ballocs,\types) = \sum_i \typei(\balloci)$. Given fixed type profile $\types$, we write $SW_{opt}$ for $\max_{\ballocs}\{SW(\ballocs,\types)\}$, and $SW_\alg(\decls) = \sum_i \typei(\alg(\decls))$.  When $D = (\decls^1, \decls^2, \dotsc, \decls^T)$ is a sequence of valuation profiles, we write $SW_\alg(D) = \frac{1}{T}\sum_t SW_\alg(\decls^t)$ for the average welfare obtained over all declarations in $D$.  We will sometimes replace subscript $\alg$ by $\mech$, in which case the social welfare is for the allocation rule of $\mech$.  Note that algorithm $\alg$ is a $c$-approximation if $SW_\alg(\types) \geq \frac{1}{c}SW_{opt}$ for all $\types$.

Given allocation rule $\alg$, agent $i$, declaration profile $\declsmi$, and set $S$, the \emph{critical price} $\th_i^\alg(S,\declsmi)$ for $S$ is the minimum value that agent $i$ could bid on set $S$ and be allocated $S$ by $\alg$ given fixed $\declsmi$.  That is, $\th_i(S,\declsmi) = \inf\{ v : \exists \decli, \decli(S) = v, \alg_i(\decli, \declsmi) = S\}.$



We say that a declaration $\decli$ is \emph{weakly dominated} by declaration $\decli'$ for agent $i$ if $\utili(\decli,\declsmi) \leq \utili(\decli',\declsmi)$ for all $\declsmi$, and $\utili(\decli,\declsmi) < \utili(\decli',\declsmi)$ for some $\declsmi$.  


Declaration sequence $D = (\decls^0, \decls^1, \dotsc, \decls^T)$ \emph{minimizes external regret} for agent $i$ if, for any fixed declaration $\decli$, $\sum_t \utili(\decli^t,\declsmi^t) \geq \sum_t \utili(\decli,\declsmi^t) + o(T)$.  That is, the utility of agent $i$ approaches the utility of the optimal fixed strategy in hindsight.



%
%
Declaration sequence $D = (\decls^0, \decls^1, \dotsc, \decls^T)$ is an instance of \emph{response dynamics} if, for all $1 \leq t \leq T$, profiles $\decls^{t-1}$ and $\decls^t$ differ on the declaration of at most one player.  Response dynamics $D$ is an instance of \emph{best-response dynamics} if, whenever $\decls^{t-1}$ and $\decls^t$ differ on the declaration of agent $i$, $\decl_i^t$ maximizes agent $i$'s utility given the declarations of the other bidders. That is, $\decl_i^t \in \argmax_{\decl}\{\utili(\decl,\declsmi^t)\}$.

\section{Regret-Minimizing Bidders}

In this section we prove that if agents avoid algorithmically dominated strategies and minimize external regret, then a loser-independent monotone algorithm $\alg$ can be converted into a mechanism with almost no loss to its average approximation ratio over many rounds.  The mechanism, $\mech_\alg$, is described in Figure \ref{fig.mech.regret}.
%
%
Mechanism $\mech_\alg$ proceeds by first simplifying the declaration given by each agent, then passing the simplified declarations to algorithm $\alg$.  The resulting allocation is paired with a payment scheme that charges critical prices.  

\begin{figure}
\begin{center}
	\fbox{
		\begin{minipage}{0.95\linewidth}

\textbf{Mechanism} $\calm_\alg$:
\medskip
\hrule

\medskip

\textbf{Input:} Declaration profile $\decls = \decl_1, \dotsc, \decl_n$.

%
\begin{tabbing}
1. \= $\decl' \leftarrow $ \texttt{SIMPLIFY}$(\decls)$. \\
2. \> Allocate $\alg(\decls')$, charge critical prices.
\end{tabbing}
		\end{minipage}
	}
	\fbox{
		\begin{minipage}{0.95\linewidth}

\textbf{Procedure} \texttt{SIMPLIFY:}
\medskip
\hrule

\medskip
\textbf{Input:} Declaration profile $\decls = \decl_1, \dotsc, \decl_n$.

%
\begin{tabbing}
1. \= For each $i \in [n]$: \\
2. \> \quad \= Choose $S_i \in \argmax_S\{\decli(S)\}$, breaking \\ 
\>\> ties in favour of smaller sets. \\
3. \> \> $\decli' \leftarrow (S_i,\decli(S_i))$. \\
4. \> Return $(\decl_1', \dotsc, \decl_n')$.
\end{tabbing}
		\end{minipage}
	}

			\caption{Mechanism for regret-minimizing bidders, based on monotone algorithm $\alg$.  Uses subprocedure \texttt{SIMPLIFY}.}
			\label{fig.mech.regret}
	\end{center}
\end{figure}

The simplification process \texttt{SIMPLIFY} essentially converts any declaration into a single-minded declaration (and does not affect declarations that are already single-minded).  We will therefore assume without loss of generality that agents make single-minded declarations, as additional information is not used by the mechanism.\footnote{We note, however, that this is not the same as assuming that agents are single-minded; our results hold for bidders with general private valuations.} 

Fix a particular combinatorial auction problem and type profile $\types$, and let $\alg$ be some monotone approximation algorithm.  Let $\decls$ be a declaration profile; we suppose each $\decli$ is a single-minded bid for set $S_i$.
%
We draw the following conclusion about the bidding choices of rational agents. 

\begin{lem}
\label{lem.undom1}
Declaration $\decli$ is an undominated strategy for agent $i$ if and only if $\decli(S_i) = \typei(S_i)$.  
\end{lem}
\begin{proof}
For all $\declsmi$, $\mech_\alg(\decli,\declsmi)$ either allocates $S_i$ or $\emptyset$ to agent $i$.  
Thus agent $i$'s utility for declaring $\decli$, $\utili(\decli,\declsmi)$, is $\typei(S_i) - \crit_i^{\mech_\alg}(S_i,\declsmi)$ when $\decli(S_i) > \crit_i^{\mech_\alg}(S_i,\declsmi)$, and $0$ otherwise.  A declaration of $\decli(S_i) = \typei(S)$ therefore maximizes $\utili(\decli,\declsmi)$ for all $\declsmi$.

On the other hand, if $\decli(S_i) \neq \typei(S_i)$, let $\decli'$ be the single-minded declaration for $S_i$ at value $\typei(S_i)$.  Then for any $\declsmi$ such that $\criti^\alg(S_i,\declsmi)$ lies between $\decli(S_i)$ and $\typei(S_i)$, $\utili(\decli',\declsmi) > \utili(\decli,\declsmi)$.
For simplicity we will assume such a $\declsmi$ exists; handling the general case requires only a technical and uninteresting extension of notation\footnote{If $\criti^\alg(S_i,\declsmi)$ never lies between $\decli(S_i)$ and $\typei(S_i))$ for any $\declsmi$, then $\mech_\alg(\decli,\declsmi) = \mech_\alg(\decli',\declsmi)$ for all $\declsmi$, so $\decli$ and $\decli'$ are equivalent strategies.  We can therefore think of $\decli$ as being ``the same'' as a single-minded declaration for $S_i$ at value $\typei(S_i)$.  We will ignore this technical issue for the remainder of the paper, in the interest of keeping the exposition simple.}.  Thus declaration $\decli'$ weakly dominates declaration $\decli$.  
%
\end{proof}

One implication of Lemma \ref{lem.undom1} is that the strategic choice of an agent participating in mechanism $\calm_\alg$ reduces to a linear optimization problem.  On each round, we can think of agent $i$ as choosing set $S_i$, which is the set he will attempt to win that round.  Once $S_i$ is chosen, an undominated declaration for agent $i$ is determined: the single-minded declaration for $S_i$ at value $\typei(S_i)$.  Given that agent $i$ chooses set $S_i$, his utility will be $\typei(S_i)-w_i$, where $w_i = \min\{\typei(S_i),\crit_i^\alg(S_i,\declsmi)\}$ is the price for set $S_i$, determined by the declarations of the other agents, capped at $\typei(S_i)$.  Thus, since utilities are linear in the choice of $S_i$, agents can indeed apply the regret-minimization algorithm of Kalai and Vempala \cite{KV-05} to choose strategies that minimize external regret.


We now proceed with bounding the social welfare obtained by $\calm_\alg$.  Let $A_1, \dotsc, A_n$ be an optimal assignment for types $\types$.  Suppose that $D = \decls^1, \dotsc, \decls^T$ is a sequence of declarations to our mechanism.
The definition of regret minimization then immediately implies the following. 

\begin{lem}
\label{lem.minregret}
If agent $i$ minimizes his external regret in bid sequence $D$, then $\frac{1}{T}\sum_t (\typei(\alg(\decls^t)) + \criti^\alg(A_i,\declsmi^t)) \geq \typei(A_i) - o(1)$.
\end{lem}

Assume now that algorithm $\alg$ is loser independent.  We can then relate the value of the solution returned by an algorithm to the critical prices of the sets in an optimal solution.

\begin{lem}
\label{lem.loserindep}
If $\alg$ is a monotone loser-independent $c$-approximate algorithm,  then $\sum_i \decli(\alg(\decls)) \geq \frac{1}{c}\sum_i \crit_i^\alg(A_i,\declsmi)$.
\end{lem}
\begin{proof}
Choose some $\eps > 0$.  For each $i$, let $\decli'$ be the pointwise maximum between $\decli$ and the single-minded declaration for set $A_i$ at value $\crit_i^{\alg}(A_i,\declsmi) - \eps$.  The loser independence property implies that we can perform this operation independently for each agent (i.e. without affecting critical prices), and moreover $\alg(\decls') = \alg(\decls)$.  Since $\alg$ is a $c$-approximate algorithm, $\sum_i \decli'(\alg(\decls')) \geq \frac{1}{c}\sum_i \decli'(A_i) \geq \frac{1}{c}\sum_i (\crit_i^\alg(A_i,\declsmi)-\eps)$.  Additionally, since $\decli'(T) = \decli(T)$ whenever $\decli(T) \geq \crit_i^\alg(T,\declsmi)$ (from the definition of $\decli'$), we have $\decli'(\alg(\decls')) = \decli(\alg(\decls))$ for all $i$.  We conclude that $\sum_i \decli(\alg(\decls)) \geq \frac{1}{c}\sum_i (\crit_i^\alg(A_i,\declsmi)-\eps)$ for all $\eps > 0$, as required.
\end{proof}

We are now ready to proceed with the proof of our main result in this section.  

\begin{thm}
\label{thm.regretmin}
Any monotone loser-independent $c$-approximate algorithm can be implemented as a mechanism with $c + 1$ price of total anarchy.
\end{thm}
\begin{proof}
Let $D = \decls^1, \dotsc, \decls^T$ be a sequence of declarations in which all agents minimize external regret.
By Lemma \ref{lem.minregret}, $\frac{1}{T}\sum_t (\typei(\alg(\decls^t)) + \crit_i^\alg(A_i,\declsmi^t)) \geq \typei(A_i) - o(1)$.  Summing over all $i$, we have $\frac{1}{T}\sum_t\sum_i (\typei(\alg(\decls^t)) + \crit_i^\alg(A_i,\declsmi^t)) \geq SW_{OPT} - (n)(o(1))$.  By Lemma \ref{lem.loserindep}, this implies $\frac{1}{T}\sum_t\sum_i (\typei(\alg(\decls^t)) + c\decl_i^t(\alg(\decls^t))) \geq SW_{OPT} - (n)(o(1))$.  We know $\decl_i^t(\alg(\decls^t)) = \typei(\alg(\decls^t))$ for all $i$ and $t$ by Lemma \ref{lem.undom1}, so we conclude $(c+1)\frac{1}{T}\sum_t\sum_i \typei(\alg(\decls^t)) \geq SW_{OPT} - (n)(o(1))$.  Since the term hidden by the asymptotic notation vanishes with $T$ and does not depend on $n$, we obtain the desired result.
\end{proof}

Theorem \ref{thm.regretmin} is very general, as it applies to a number of different problem settings in which loser-independent monotone approximation algorithms are known.  As a few particular examples, Theorem \ref{thm.regretmin} yields an $O(\sqrt{m})$ implementation of the combinatorial auction problem \cite{LOS-99}, an $s+1$ implementation of the combinatorial auction problem where sets are restricted to cardinality $s$ (using a simple greedy algorithm), an $O(m^{1/(B-1)})$ implementation of the unsplittable flow problem with minimum edge capacity $B$ \cite{BKV-05}, and an $O(R^{4/3})$ implementation of the combinatorial auction of convex bundles in the plane where $R$ is the maximum aspect ratio over all desired bundles \cite{BB-04}.

We note that, since agents experience no regret at a pure Nash equilibrium, an immediate corollary to Theorem \ref{thm.regretmin} is that any monotone loser-independent $c$-approximate algorithm can be implemented as a mechanism with $c + 1$ price of anarchy.  We remark that an alternative proof of this result has been given recently using a different mechanism construction \cite{BL-10}.

Also, the rate at which the welfare obtained by $\calm_\alg$ converges to an average that is a $c+1$ approximation to optimal depends on the rate of convergence of players' external regret to $0$.  The average welfare obtained after $T$ rounds will have an additive loss of $(n)(r(T))$, where $r(T)$ is the average regret experienced by an agent after $T$ rounds.  Assuming that agents apply algorithms that minimize regret at a rate of $r(T) = o(1/\sqrt{T})$, which is easily attainable using the algorithm of Kalai and Vempala \cite{KV-05}, the additive error term is at most a constant when $T$ is at least quadratic in $n$.

\subsection{Resilience to Byzantine Agents}

Suppose that in addition to regret-minimizing agents, the auction participants include byzantine agents.  The only restriction we impose on the behaviour of such agents is that they do not overbid on any set; that is $\decli(S) \leq \typei(S)$ for any $S$ and byzantine agent $i$.  We can motivate this restriction either through our characterization of undominated strategies in Lemma \ref{lem.undom1}, or by thinking of byzantine players as not understanding how to participate rationally in the auction, and hence likely to be conservative in the way that they participate.  Under this assumption, since Lemma \ref{lem.loserindep} holds for any declaration profile, we easily obtain the following generalization to Theorem \ref{thm.regretmin}.

\begin{prop}
\label{prop.resil}
Suppose $\alg$ is a monotone loser-independent $c$-approximate algorithm and $D$ is a declaration sequence for $\mech_\alg$.  If $N \subseteq [n]$ is a collection of agents that minimize regret in $D$, and the remaining agents never bid more than their true values on any set in $D$, then $\frac{1}{T}\sum_t SW_\alg(\decls^t) \geq \frac{1}{c+1}\sum_{i \in N}SW_{opt} + |N|(o(1))$.
\end{prop}

\subsection{Importance of Loser-Independence}

We note that the loser independence property is necessary for Theorem \ref{thm.regretmin}, as the following example demonstrates.

\begin{example}
Consider an auction problem in which no agent can be allocated more than $s$ objects, and moreoever $M = A \cup B$ where $|A| = |B| = m/2$ and the mechanism must either allocate objects in $A$ or objects in $B$, but not both.  Consider the algorithm that takes the maximum over two solutions: a greedy assignment of subsets of $A$, and a greedy assignment of subsets of $B$.  This algorithm obtains an $s+1$ approximation.

Consider now an instance of the problem in which a single agent desires all of $B$ with value $1$, and each of $m/2$ agents desires a separate singleton in $A$ with value $1-\eps$.  Suppose that the agent desiring $B$ declares his valuation truthfully, but the other agents declare the zero valuation.  On this input, the algorithm under consideration obtains only an $m/2$ approximation to the optimal solution.  However, this set of declarations forms a Nash equilibrium, and hence each agent has zero regret under this input profile.  Thus, even if agents minimize their regret, our mechanism may obtain a very poor approximation to the optimal social welfare over arbitrarily many auction rounds.  
\end{example}

\section{Best-Response Agents}

In this section we consider the problem of designing mechanisms for agents that apply myopic best-response strategies asynchronously.  Recall that in our model agents are chosen for update uniformly at random, one per round.  In order to keep our exposition clear, we will make two additional assumptions about the nature of the best-response behaviour (which can be removed, as we discuss below).  First, we will suppose that in the initial state every bidder makes the empty declaration $\emptyset$.  Second, we suppose that if a bidder is chosen for update but cannot improve his utility, he will choose to maintain his previous strategy.  These assumptions will simplify the process of characterizing best-response strategies of agents, and in particular the statement of Lemma \ref{lem.separated} in the next section.  It is possible to remove these assumptions, at the cost of a minor modification to the mechanisms we propose.  We defer a more complete discussion to the appendix.

A simple example shows that mechanism $\calm_\alg$ may not converge to a Nash equilibrium via best-response dynamics; this example is presented in the appendix.  We conjecture that, \emph{on average}, the best-response dynamics on mechanism $\calm_\alg$ obtains a good approximation to the optimal social welfare.  

\begin{conj}
\label{conj.best.response}
If $\alg$ is a monotone loser-independent $c$-approximate algorithm, then $\calm_\alg$ has $O(c)$ price of (myopic) sinking.
\end{conj}

As partial progress toward resolving Conjecture \ref{conj.best.response}, we construct alternative mechanisms that are more amenable to best-response analysis.  These mechanisms are tailored specifically to the general combinatorial auction problem, and combinatorial auctions with cardinality-restricted sets. Our hope is to demonstrate the intuition behind Conjecture \ref{conj.best.response} and explore mechanism design tools that may prove useful in its resolution.

The primary tool we will use is the following probabilistic lemma, which pertains to any mechanism in a best-response setting.  Suppose $\calm$ is a mechanism, and $D$ is a sequence of best-response declarations for $\calm$.
For any $\decls$, let $P_1(\decls) = P_1(\declsmi)$ be some property of $\decls$ that does not depend on $\decli$, and let $P_2(\decls) = P_2(\decli)$ be some property depending only on $\decli$.

\begin{lem}
\label{lem.numbad}
Suppose that, for any $\decls$, if $P_1(\declsmi)$ is false, then any best response by agent $i$, $\decli$, satisfies $P_2(\decli)$.  Then for all $\eps > 0$, if best-response dynamics is run for $T > \eps^{-1}n$ steps, there will be at least $(\frac{1}{2}-\eps)T$ steps $t$ for which either $P_1(\declsmi^t)$ or $P_2(\decli^t)$ is true, with probability at least $1 - e^{-T\eps^2/32n}$.
\end{lem}
\begin{proof}[Proof (sketch)]
Let $B_i^t$ be the event that neither $P_1(\declsmi^t)$ nor $P_2(\decli^t)$ is true, and 
let $A_i^t$ denote the event that $P_2(\decli^t)$ is true.  Our goal is to bound the number of occurrances of $B_i^t$.  

Note that if $B_i^t$ occurs and agent $i$ is chosen for update on step $t+1$, then $A_i^{t+1}$ occurs (by assumption).  Alternatively, if $A_i^t$ occurs but agent $i$ is not chosen for update on step $t+1$ then $A_i^{t+1}$ occurs, since $A_i$ depends only on the declaration of agent $i$. Thus events $A_i^t$ and $B_i^t$ can be compared to a random walk on $\{0,1\}$, where at each step the current state changes with probability $1/n$.  The number of occurrances of $B_i^t$ is dominated by the number of occurrances of $0$ in such a random walk.  A straightforward application of the method of bounded average differences shows that this value is concentrated around its expectation, which is at most $\frac{T}{2} + \frac{n}{2}$.  Thus, as long as $T > \eps^{-1}n$, the number of occurrances of $B_i^t$ will be concentrated at $T(\frac{1}{2}+\eps)$, giving the desired bound.  Additional details are deferred to the appendix. 
%
\end{proof}

\subsection{A Mechanism for $s$-CAs}
\label{sec.mech.sCA}

Consider the $s$-CA problem, which is a combinatorial auction in which no agent can be allocated more than $s$ objects.  An algorithm that greedily assigns sets in descending order by value obtains an $(s+1)$ approximation.\footnote{And an $s$ approximation for single-minded declarations.}  Call this algorithm $\alg_{sCA}$.  We will construct a mechanism $\calm_{sCA}$ based on $\alg_{sCA}$; it is described in Figure \ref{fig.mech.sCA}.  This algorithm simplifies incoming bids (in the same way as $\calm_\alg$) and runs algorithm $\alg_{sCA}$ to find a potential allocation.  However, an additional condition for inclusion in the solution is imposed: the value declared for a set must be larger than the sum of all bids for intersecting sets.  Potential allocations that satisfy this condition are allocated, and the mechanism charges critical prices (that is, the smallest value at which an agent would be allocated their set by $\calm_{sCA}$, which is not necessarily the same as the critical price for $\alg_{sCA}$).

\begin{figure}
\begin{center}
	\fbox{
		\begin{minipage}{0.9\linewidth}

\textbf{Mechanism} $\calm_{sCA}$:
\medskip
\hrule

\medskip

\textbf{Input:} Declaration profile $\decls = \decl_1, \dotsc, \decl_n$.

%
\begin{tabbing}
1. \= $\decl' \leftarrow $ \texttt{SIMPLIFY}$(\decls)$, say $\decli' = (S_i,v_i)$ \\
2. \> $(T_1, \dotsc, T_n) \leftarrow \alg_{sCA}(\decls')$.\\
3. \> For each $i$ such that $T_i \neq \emptyset$:\\
4. \> \quad \= $R \leftarrow \{j : S_j \cap T_i \neq \emptyset \}$. \\
5. \> \> $p_i \leftarrow \sum_{j \in R}\decl_j(S_j)$. \\
6. \> \> If $\decli'(T_i) \leq p_i$, set $T_i \leftarrow \emptyset$, $p_i \leftarrow 0$.\\
7. \> Allocate $T_1, \dotsc, T_n$, charge critical prices.
\end{tabbing}
		\end{minipage}
	}
			\caption{Mechanism $\calm_{sCA}$, an implementation of greedy algorithm $\alg_{sCA}$ for the $s$-CA problem.}
			\label{fig.mech.sCA}
	\end{center}
\end{figure}

We note that since our mechanism implements a monotone algorithm and charges critical prices, Lemma \ref{lem.undom1} implies that undominated strategies for agent $i$ involve choosing a set $S_i$ and making a single-minded bid for $S_i$ at value $\typei(S_i)$.  We will therefore assume that agents bid in this way.

Suppose that $\decls$ is a declaration profile, where each $\decli$ is single-minded for $S_i$.  For any set $T$, define $R_i(\decls,T) = \{j : j \neq i, S_j \cap T \neq \emptyset\}$.  We also define $Q_i(\decls,T) = \{j : j \in R_i(\decls,T), \decl_j(S_j) < \typei(T) \}$.  That is, $R_i$ is the set of bidders other than $i$ whose single-minded declared sets intersect $T$, and $Q_i$ is the subset of those bidders whose single-minded declared values are less than agent $i$'s true value for $T$.  We then say that $\decls$ is \emph{separated for agent $i$} if $\sum_{j \in Q_i(\decls,S_i)}\decl_j(S_j) \leq \decli(S_i)$ and $\decls$ is \emph{separated} if it is separated for every bidder.  Since an agent gains positive utility only if the declaration is separated for him, and since the intial state is the empty declaration profile (which is separated), we draw the following conclusion.

\begin{lem}
\label{lem.separated}
At each step of the best-response dynamics for mechanism $\calm_{sCA}$, the declaration profile submitted by the agents will be separated.
\end{lem}

For the remainder of the section we will assume that declaration profiles are separated.  Under this assumption, the behaviour of mechanism $\calm_{sCA}$ simplifies in a fortuitous way.

\begin{prop}
\label{prop.sep.mech.sCA}
If $\decls$ is separated, then $\mech_{sCA}$ allocates $S_i$ to agent $i$ precisely when $\decli(S_i) > \max_{j \in R_i(S_i,\decls)}\decl_j(S_j)$.
\end{prop}
\begin{proof}
If $\decli(S_i) > \max_{j \in R_i(S_i,\decls)}\decl_j(S_j)$, then $S_i$ is allocated by $\alg_{sCA}(\decls)$.  Furthermore $Q_i(S_i,\decls) = R_i(S_i,\decls)$, so $\decli(S_i) > \sum_{j \in R_i(S_i,\decls)}\decl_j(S_j)$ implies $\decli(S_i) > \sum_{j \in R_i(S_i,\decls)}\decl_j(S_j)$ and hence $S_i$ will be allocated by $\mech_{sCA}$.  On the other hand, if $\decli(S_i) \leq \max_{j \in R_i(S_i,\decls)}\decl_j(S_j)$, then certainly $\decli(S_i) \leq \sum_{j \in R_i(S_i,\decls)}\decl_j(S_j)$ so $S_i$ is not allocated by $\mech_{sCA}$.
\end{proof}

Let $A_1, \dotsc, A_n$ be an optimal allocation with respect to the agents' true types $\types$.

\begin{prop}
\label{prop.sep.response}
If $\decls$ is separated and $\sum_{j \in R_i(\decls,A_i)}\decl_j(S_j) < \frac{1}{2}\typei(A_i)$, then any utility-maximizing declaration for agent $i$, $\decli$, will be a single-minded declaration for some $S_i$ with $\decli(S_i) \geq \frac{1}{2}\typei(A_i)$.
\end{prop}

For declaration profile $\decls$, let $G$ denote the set of agents $i$ for which either $\sum_{j \in R_i(\decls,A_i)}\decl_j(S_j) > \frac{1}{2}\typei(A_i)$ or  $\decli(S_i) \geq \frac{1}{2}\typei(A_i)$.  We now bound the social welfare obtained by $\calm_{sCA}$ with respect to the optimal assignment to agents in $G$.

\begin{lem}
\label{lem.goodevents.welfare}
$$SW_{\mech_{sCA}}(\decls) \geq \frac{1}{4(s+1)}\sum_{i \in G}\typei(A_i).$$
\end{lem}

We are now ready to bound the average social welfare of our mechanism, over sufficiently many rounds, with respect to the approximation factor of algorithm $\alg$.

\begin{thm}
\label{thm.best-response.1}
Choose $\eps > 0$ and suppose $D = d^1, \dotsc, d^T$ is an instance of best-response dynamics with random player order, where agents play undominated strategies, and $T > \eps^{-1}n$.  Then 
\[ SW_{\mech_{sCA}}(D) \geq \left(\frac{1}{8(s+1)}-\eps\right)SW_{opt}(\types) \]
with probability at least $1-ne^{-T\eps^2/32n}$.
\end{thm}
\begin{proof}
Let $G_t$ be the set of agents $G$ from Lemma \ref{lem.goodevents.welfare} on step $t$ (i.e. with respect to declaration $\decls^t$).  Lemma \ref{lem.numbad} and Proposition \ref{prop.sep.response} together imply that each agent $i$ will be in $G_t$ for at least $(\frac{1}{2}-\eps)T$ values of $t$, with probability at least $1-e^{-T\eps^2/32n}$.  The union bound then implies that this occurs for every agent with probability at least $1-ne^{-T\eps^2/32n}$.  Conditioning on the occurrance of this event, Lemma \ref{lem.goodevents.welfare} implies
\begin{align*}
SW_{\mech_{sCA}}(D) & = \frac{1}{T}\sum_t SW_{\mech_{sCA}}(\decls^t) \\
& \geq \frac{1}{4(s+1)T}\sum_t\sum_{i \in G_t} \typei(A_i) \\
& \geq \frac{1}{4(s+1)T}\sum_i\left(\frac{T}{2}-\eps\right)\typei(A_i) \\
& \geq \left(\frac{1}{8(s+1)}-\eps\right)SW_{opt}(\types)
\end{align*}
which implies the required bound.  
\end{proof}

Taking, say, $\eps = 0.1$, and assuming $T >> n$, we conclude that $SW_{\mech_{sCA}}(D) > \frac{1}{O(s)}SW_{opt}(\types)$ with high probability.  Thus $\calm_{sCA}$ implements an $O(s)$ approximation to the $s$-CA problem for best-response bidders.

\subsection{A Mechanism for General CAs}
\label{sec.mech.CA}

Consider the following algorithm for the general CA problem: try greedily assigning sets, of size at most $\sqrt{m}$, by value; return either the resulting solution or the allocation that gives all items to a single agent, whichever has higher welfare.  This algorithm is an $O(\sqrt{m})$ approximation \cite{MN-08}.  We will construct a mechanism $\calm_{CA}$ based on this algorithm; it is described in Figure \ref{fig.mech.CA}.  $\calm_{CA}$ essentially implements two copies of $\calm_{sCA}$: one for sets of size at most $\sqrt{m}$ (which we will call $\calm_{\sqrt{m}CA}$), and one for allocating all objects to a single bidder; it then takes the maximum of the two solutions.  We add one additional modification: with vanishingly small probability $\gamma$, $\calm_{CA}$ ignores bids for $M$ and behaves as $\calm_{\sqrt{m}CA}$.  The purpose of this modification is to encourage agents to bid on small sets, even when the presence of a high-valued bid for a large set would seem to indicate that bidding on small sets is fruitless.  



\begin{figure}
\begin{center}
	\fbox{
		\begin{minipage}{0.9\linewidth}

\textbf{Mechanism} $\calm_{CA}$:

\hrule

\medskip

\textbf{Input:} Declaration profile $\decls = \decl_1, \dotsc, \decl_n$.

%
\begin{tabbing}
1. \= $\decl' \leftarrow $ \texttt{SIMPLIFY}$(\decls)$, say $\decli' = (S_i,v_i)$ \\
2. \> With probability $\gamma$:\\
3. \> \quad \= For all $i$ with $S_i = M$, $\decli' \leftarrow \emptyset$. \\
4. \> Let $(T_1, \dotsc, T_n) \leftarrow \mech_{\sqrt{m}CA}(\decls')$.\\
5. \> If $\exists i : S_i = M$: \\
6. \> \> Let $j \leftarrow \argmax_j \{\decl_j'(M) : S_j = M\}$. \\
7. \> \> If $\decl_j'(S_j) > \sum_{k \neq j: S_k = M}\decl_k'(S_k)$ and\\
\>\> \ \ $\decl_j'(S_j) > \sum_i \decl_i'(T_i)$: \\
8. \> \> \quad \= Set $T_j \leftarrow M$, $T_i \leftarrow \emptyset$ for all $i \neq j$ \\
9. \> Allocate $T_1, \dotsc, T_n$, charge critical prices.
\end{tabbing}

		\end{minipage}
	}
			\caption{Mechanism $\calm_{CA}$, a best-response implementation of a greedy algorithm for the CA problem.  Parameter $\gamma > 0$ is an arbitrarily small positive constant.  Note $\calm_{\sqrt{m}CA}$ is $\calm_{sCA}$ from Figure \ref{fig.mech.sCA} with $s = \sqrt{m}$.}
			\label{fig.mech.CA}
	\end{center}
\end{figure}

The analysis of the average social welfare obtained by $\calm_{CA}$ closely follows the analysis for $\calm_{sCA}$.  Our high-level approach is to apply this analysis twice: once for allocations of sets of size at most $\sqrt{m}$, and once for allocations of all objects to a single bidder.  The primary complicating factor is that the bidding choice of an agent may be influenced by the mechanism's choice of whether or not to allocate $M$ to a single bidder; this can be handled by a careful analysis of utility-maximizing declarations.  We defer all details to the appendix.  The final result is the following.

\begin{thm}
\label{thm.best-response.2}
Choose $\eps > 0$ and suppose $D = d^1, \dotsc, d^T$ is an instance of best-response dynamics with random player order, where agents play undominated strategies, and $T > \eps^{-1}n$.  Then 
\[ SW_{\mech_{CA}}(D) \geq \left(\frac{1}{O(\sqrt{m})}-\eps\right)SW_{opt}(\types) \]
with probability at least $1-2ne^{-T\eps^2/32n}$.
\end{thm}

We conclude that mechanism $\calm_{CA}$ implements an $O(\sqrt{m})$ approximation to the combinatorial auction problem for best-response bidders, with high probability, whenever $T >> n$.

\section{Conclusions and Future Work}

We considered the problem of designing mechanisms for use with regret-minimizing and best-response bidders in repeated combinatorial auctions.  
We presented a general black-box construction for the regret-minimization model, which implements any monotone loser-independent approximation algorithm.  For the best-response model, we constructed an $O(\sqrt{m})$-approximate mechanism for the combinatorial auction problem. 

The most obvious direction for future research is to extend our results to implement additional algorithms.  Our best-response mechanisms made specific use of the structure of greedy CA algorithms, but it seems likely that our approach can be generalized.  Another specific question of note is whether Conjecture \ref{conj.best.response} is true, and the mechanism we proposed for regret-minimizing bidders also yields good performance when used by best-response bidders.  Our techniques appear limited to loser-independent algorithms; can algorithms that are not loser-independent be implemented in regret-minimization or best-response domains?  A broader research topic is to explore other models for reasonable bidder behaviour, which may admit different mechanism implementations.


\appendix

\section*{Appendix A: Proof of Proposition \ref{prop.resil}}
\renewcommand{\thesection}{A}
\label{app.resil}


Recall the statement to be proven.  We suppose $\alg$ is a monotone loser-independent $c$-approximate algorithm and $D$ is a declaration sequence for $\mech_\alg$.  If $N \subseteq [n]$ is a collection of agents that minimize regret in $D$, and the remaining agents never bid more than their true values on any set in $D$, then we must show that $\frac{1}{T}\sum_t SW_\alg(\decls^t) \geq \frac{1}{c+1}\sum_{i \in N}\typei(A_i) + |N|(o(1))$ for any allocation profile $A_1, \dotsc, A_n$.

We proceed with the proof.
By Lemma \ref{lem.minregret} and summing over all $i \in N$, 
\begin{align*}
& \sum_{i \in N}\typei(A_i) - |N|(o(1)) \\ 
\leq & \frac{1}{T}\sum_t\sum_{i\in N} (\typei(\alg(\decls^t)) + \criti(A_i,\declsmi^t)) \\
\leq & \frac{1}{T}\sum_t\sum_{i\in [n]} (\typei(\alg(\decls^t)) + \criti(A_i,\declsmi^t)).
\end{align*}  
By Lemma \ref{lem.loserindep}, this implies $\frac{1}{T}\sum_t\sum_i (\typei(\alg(\decls^t)) + c\decl_i^t(\alg(\decls^t))) \geq \sum_{i \in N}\typei(A_i) - |N|(o(1))$.  The result then follows from the fact that $\decl_i^t(\alg(\decls^t)) \leq \type_i^t(\alg(\decls^t))$ for all $i$.

\section*{Appendix B: Assumptions on the Best-Response Model}
\label{app.assume}

Recall that in our model of best-response dynamics, we assumed that in the initial state every bidder makes the empty declaration $\emptyset$, and that if a bidder is chosen for update but cannot improve his utility, he will choose to maintain his previous strategy.  We used these assumptions to argue that agents make only separated declarations when participating in mechanisms $\calm_{sCA}$ and $\calm_{CA}$.  

These assumptions can be removed, as follows.  We can modify mechanisms $\calm_{sCA}$ and $\calm_{CA}$ so that, with vanishingly small probability, an alternative allocation rule is used.  This alternative rule chooses an agent at random, and assigns him all objects at no cost \emph{as long as the input declaration is separated for that agent}.  Thus, any separated declaration by agent $i$ results in positive expected utility.  Then, since any non-separated declaration by an agent results in a utility of $0$ for that agent, it must be that the utility-maximizing declaration by any agent must be a separated declaration.

It follows that after each bidder is chosen at least once for update, and every step thereafter, the input declaration will be separated.  Thus, with high probability, every declaration after $O(n\log n)$ steps will be separated (by the coupon-collector problem).  Lemma \ref{lem.separated} will therefore hold after $O(n\log n)$ steps of best-response dynamics, with high probability; the remainder of the analysis can then proceed without change.

\section*{Appendix C: Best-response Dynamics on $\calm_\alg$ might not Converge}
\label{app.converge}

Consider a combinatorial auction problem where each agent can receive at most 2 items.  Let $\alg$ be the greedy allocation rule that allocates sets greedily by value. Suppose there are 6 agents and 4 objects, say $\{a,b,c,d\}$, and that the agents' true valuations are given by the following set of bids (where the value for a set not listed is taken to be the maximum over its subsets).

\begin{center}
\begin{tabular}{c|c|c}
player & set & value \\
\hline
$1$ & $\{a,b\}$ & $4$ \\
$1$ & $\{d\}$ & $6$ \\
$2$ & $\{a\}$ & $2$ \\
$2$ & $\{b,c\}$ & $5$ \\
$3$ & $\{c\}$ & $4$ \\
$4$ & $\{d\}$ & $5$ \\
\end{tabular}
\end{center}

Since agents $3$ and $4$ are single-minded, they always maximize their utility by declaring values truthfully, so we can assume in any sequence of best-response moves that they do so.   By contrast, players $1$ and $2$ each have a strategic choice to make each round: which of their two desired sets should they bid upon?  Note that once this decision is made, the way to bid is determined by Lemma \ref{lem.undom1} (i.e. bid truthfully for the desired set).  We will now show that from each of the resulting 4 possible declaration profiles, some player has incentive to change their declaration.

Suppose that player $1$ bids for set $\{d\}$ and player $2$ bids for $\{b,c\}$.  (Note that under the assumption that all declarations start empty, we can reach such a state by selecting agents for update in the order 3,4,1,2,1).  Then player $1$ has no incentive to change his declaration, but player $2$ would benefit by changing his bid to $\{a\}$ (increasing his utility from $1$ to $2$).  From that state, player $1$ then has incentive to remove his bid for $\{d\}$ and instead win set $\{a,b\}$, also increasing his utility from $1$ to $2$.  If he does so, player $2$'s utility becomes $0$, so player $2$ benefits by switching back to his bid for $\{b,c\}$.  This results in player $1$ winning nothing, so player $1$ gains by adding back his bid for $\{d\}$.  We have returned to the original state, and hence no reachable state forms an equilibrium.

\section*{Appendix D: Proof of Lemma \ref{lem.numbad}}
\label{app.prob}

Recall the statement to be proven.  We suppose that, for any $\decls$, if $P_1(\declsmi)$ is true, then the best response by agent $i$, $\decli$, satisfies $P_2(\decli)$.  We then wish to show that, for all $\eps > 0$, if best-response dynamics is run for $T > \eps^{-1}n$ steps, there will be at least $(\frac{1}{2}-\eps)T$ steps $t$ for which either $P_1(\declsmi^t)$ is false or $P_2(\decli^t)$ is true, with probability at least $1 - e^{-T\eps^2/32n}$.

\begin{proof}
For each $t$, let $B_i^t$ be the event that neither $P_1(\declsmi^t)$ nor $P_2(\decli^t)$ is true. 
Let $A_i^t$ denote the event that $P_2(\decli^t)$ is true.
Note that $A_i^t$ and $B_i^t$ are mutually exclusive.  Consider the steps in which either $A_i^t$ or $B_i^t$ occurs: let $C_r$ denote the event that $B_i^t$ occurs on the $r$th step in which either $A_i^t$ or $B_i^t$ occurs.
We can think of the index $r$ as representing time steps, where we skip any time step in which neither $A_i^t$ nor $B_i^t$ occurs.  We wish to show $Pr[\sum_{r \leq T} C_r > (\frac{1}{2}+\eps)T] < e^{-T\eps^2/32n}$, which implies our desired result.

Suppose that event $C_r$ occurs; this implies that $B_i^t$ occurs, where $t$ is the $r$th step on which either $A_i^t$ or $B_i^t$ occurs.  With probability $\frac{1}{n}$ agent $i$ is chosen for update on step $t+1$.  Conditioning on that event, $A_i^{t+1}$ occurs (by the assumption of the Lemma), and hence $C_{r+1}$ does not occur.  We conclude $Pr[C_{r+1}|C_r] \leq (1-1/n)$.  

Next suppose that event $C_r$ does not occur; this implies that $A_i^t$ occurs, where $t$ is the $r$th step on which either $A_i^t$ or $B_i^t$ occurs.
With probability $(1-\frac{1}{n})$ agent $i$ is not chosen for update on step $t+1$.  Conditioning on that event, $A_i^{t+1}$ occurs (since $P_2$ depends only on the declaration of agent $i$, which does not change), and hence $C_{r+1}$ does not occur.  We conclude $Pr[C_{r+1}|\neg C_r] \leq 1/n$.

Let $D_1, D_2, \dotsc, D_T$ be a random walk on $\{0,1\}$ defined by $Pr[D_r|D_{r-1}] = (1-1/n)$, $Pr[D_r|\neg D_{r-1}] = 1/n$, and initial condition $D_0$.  Then $\sum_r C_r$ is stochastically dominated by $\sum_r D_r$, and hence 
$Pr[\sum_r C_r > (\frac{1}{2}+\eps)T] \leq Pr[\sum_r D_r > (\frac{1}{2}+\eps)T]$.  It will therefore suffice to show that $Pr[\sum_{r \leq T} D_r > (\frac{1}{2}+\eps)T] < e^{-T\eps^2/32n}$.  
%

The definition of $D_r$ yields 
\begin{align*}
Pr[D_r] & = \frac{1}{n}(1-Pr[D_r]) + (1-\frac{1}{n})Pr[D_r] \\
& = \frac{1}{n} + (1-\frac{2}{n})Pr[D_r].  
\end{align*}
Solving the recurrence (with initial condition $D_0$) yields 
\begin{align*}
Pr[D_r] & = \frac{1}{2}(1-(1-2/n)^r) + D_0(1-2/n)^r \\
& = \frac{1}{2} + (D_0 - \frac{1}{2})(1-2/n)^r.
\end{align*}  
Linearity of expectation then implies 
\[ E[\sum_r D_r] = \frac{1}{2}T + (D_0 - \frac{1}{2})\frac{n}{2}(1-(1-2/n)^{T-1}).\]
From this we conclude that 
\begin{equation}
\label{eq.C.1}
E[\sum_r D_r] < \frac{1}{2}T + \frac{n}{4}
\end{equation} 
and moreover 
\begin{equation}
\begin{split}
\label{eq.C.2}
& \left|E\left[\sum_r D_r | D_0=1\right] - \right. \\
& \left. \quad\quad E\left[\sum_r D_r | D_0=0\right]\right| < \frac{n}{2}.
\end{split}
\end{equation}

Let $k = T/n$ and define random variables $F_1, \dotsc, F_k$ by $F_i = \sum_{r \in [in,(i+1)n-1]}D_r$.  Then $F_i \in [0,n]$ for all $i$, and $\sum_r D_r = \sum_r F_r$.  Furthermore, the influence of $F_i$ on $F_{i+1}, \dotsc, F_{k}$ is captured entirely by the value of $D_{(i+1)n-1}$, and from \eqref{eq.C.2} the influence of $D_{(i+1)n-1}$ on $\sum_{r = (i+1)n}^T D_r = \sum_{r=i+1}^k F_r$ is bounded by $\frac{n}{2}$.  Since the value of $F_i$ also influences the sum $\sum_r F_r$ directly by at most $n$ (due to its being included in the summation), we conclude that for all $\alpha,\alpha' \in [0,n]$,
\begin{align*}
& E\left[\sum_j F_j | F_1, \dotsc, F_{i-1}, F_i = \alpha \right] - \\
& \quad E\left[\sum_j F_j | F_1, \dotsc, F_{i-1}, F_i = \alpha' \right] \leq 3n/2
\end{align*} 
Thus, by the method of bounded average differences, we conclude that 
\begin{align*}
& Pr\left[\sum_j F_j > E\left[\sum_j F_j\right] + (\eps/2)T\right] \\
\leq & e^{-(T\eps/2)^2/2(3n/2)^2k} < e^{-T\eps^2/32n}.
\end{align*}
Since $T > \eps^{-1}n$, we have that $E[\sum_j F_j] + (\eps/2)T \leq \frac{1}{2}T + \frac{n}{4} + (\eps/2)T < (\frac{1}{2} + \frac{\eps}{2})T$.  Thus $Pr\left[\sum_j F_j > (\frac{1}{2} + \frac{\eps}{2})T\right]
< e^{-T\eps^2/32n}$, and the result follows.
\end{proof}

\section*{Appendix E: Omitted proofs from Section \ref{sec.mech.sCA}}
\label{app.sCA}
\renewcommand{\thesection}{E}
\setcounter{thm}{0}

\begin{proof}[Proof (of Lemma \ref{lem.minregret})]
Let $\decl'$ be the single-minded declaration for set $A_i$ at value $\typei(A_i)$.  From the definition of regret minimization,
\begin{align*}
\frac{1}{T}\sum_t & \utili(\decli^t,\declsmi^t) \geq \frac{1}{T}\sum_t \utili(\decl',\declsmi^t) - o(1) \\
& \geq \frac{1}{T}\sum_t \left(\typei(A_i) - \criti^\alg(A_i,\declsmi^t)\right) - o(1) \\
& = \typei(A_i) - \frac{1}{T}\sum_t \criti^\alg(A_i,\declsmi^t) - o(1).
\end{align*}
Since $\utili(\decli^t,\declsmi^t) \leq \typei(\alg(\decls^t))$ for all $t$, the result follows.
\end{proof}

\begin{proof}[Proof (of Lemma \ref{lem.separated})]

We will prove the following claim, which immediately implies the desired result due to our assumption that the initial state of best-response dynamics is the empty declaration profile (which is separated):

\begin{claim}
If $\decls$ is separated, then it remains separated after a step of the best-response dynamics.
\end{claim}

Suppose agent $i$ is chosen to update his bid, say from $\decli$ to $\decli'$.  Let $\decls' = (\decli',\declsmi)$.  If agent $i$ cannot improve his utility then $\decls' = \decls$, so $\decls'$ is separated as required.  Otherwise, he changes the set upon which he bids from, say, $S_i$ to $S_i'$.  Since $\utili(\decli',\declsmi) > 0$, it must be that $\decli(S_i) > \sum_{j \in R_i(\decls,S_i)} \decl_j(S_j)$, which implies $\decls'$ is separated for agent $i$ and $\decli(S_i) > \decl_j(S_j)$ for all $j$ such that $S_i \cap S_j \neq \emptyset$.  Hence, for all $j \neq i$, we have $Q_j(\decls',S_j) \subseteq Q_j(\decls,S_j)$. Thus, since $\decls$ is separated for all $j \neq i$, $\decls'$ must be separated for each $j \neq i$ as well.
\end{proof}


\begin{proof}[Proof (of Proposition \ref{prop.sep.response})]
Note that $\crit_i^{\mech_{sCA}}(A_i,\declsmi) = \sum_{j \in R_i(\decls,A_i)}\decl_j(S_j)$, so agent $i$ would obtain utility at least $\frac{1}{2}\typei(A_i)$ by making a single-minded declaration for set $A_i$ at value $\typei(A_i)$.  Thus his utility-maximizing declaration must make at least this much utility, and therefore is a bid for some set $S_i$ with $\typei(S_i) \geq \frac{1}{2}\typei(A_i)$.  The result then follows from Lemma \ref{lem.undom1}.
\end{proof}

\begin{proof}[Proof (of Lemma \ref{lem.goodevents.welfare})]
We begin by showing that $SW_{\mech_{sCA}}(\decls) \geq \frac{1}{2}\sum_i \decli(S_i)$.  Recall that if set $S_i$ is allocated by $\mech_{sCA}$, then $\decli(S_i) \geq \sum_{j \in R_i(\decls,S_i)}\decl_j(S_j)$.  Let $N \subseteq [n]$ be the set of agents that receive non-empty sets in $\calm_{sCA}(\decls)$.  Proposition \ref{prop.sep.mech.sCA}
implies that for all $j \in [n]$ there is some $i \in N$ such that $S_i$ intersects $S_j$. Then
\begin{align*}
\sum_i \decli(S_i) & = \sum_{i \in N}\decli(S_i) + \sum_{i \not\in N}\decli(S_i) \\
& \leq SW_{\mech_{sCA}}(\decls) + \sum_{i \in N}\sum_{j \in R_i(\decls,S_i)}\decl_j(S_j) \\
& \leq SW_{\mech_{sCA}}(\decls) + \sum_{i \in N}\decli(S_i) \\
& = 2SW_{\mech_{sCA}}(\decls)
\end{align*}
and therefore $SW_{\mech_{sCA}}(\decls) \geq \frac{1}{2}\sum_i \decli(S_i)$.

Now let $G_1$ denote the set of agents for which $\decli(S_i) \geq \frac{1}{2}\typei(A_i)$, and let $G_2$ denote the set of agents for which $\sum_{j \in R_i(\decls,A_i)}\decl_j(S_j) > \frac{1}{2}\typei(A_i)$, so that $G = G_1 \cup G_2$.

We note that $SW_{\mech_{sCA}}(\decls) \geq \frac{1}{2}\sum_i \decli(S_i) \geq \frac{1}{2}\sum_{i \in G_1} \decli(S_i) \geq \frac{1}{4}\sum_{i \in G_1} \typei(A_i)$.
Also,
$\sum_{i \in G_2}\sum_{j \in R_j(\decls,S_i)}\decl_j(S_j) \geq \frac{1}{2}\sum_{i \in G_2}\typei(A_i)$.  Since $S_j$ can intersect at most $|S_j| \leq s$ sets $A_i$, we conclude $s\sum_j \decl_j(S_j) \geq \frac{1}{2}\sum_{i \in G_2}\typei(A_i)$.  This implies $(4s)SW_{\mech_{sCA}}(\decls) \geq \sum_{i \in G_2} \typei(A_i)$.

We conclude $(4s+4)SW_{\mech_{sCA}}(\decls) \geq \sum_{i \in G_1} \decli(S_i) + \sum_{i \in G_2} \decli(S_i)$, which implies the desired result.
\end{proof}

\section*{Appendix F: Proof of Theorem \ref{thm.best-response.2}}
\label{app.CA}
\renewcommand{\thesection}{F}
\setcounter{thm}{0}

Recall the definition of a separated declaration profile from our analysis of $\calm_{sCA}$.  Lemma \ref{lem.undom1} and Lemma \ref{lem.separated} apply to $\calm_{CA}$ for the same reasons as $\calm_{sCA}$, so we will assume that all declarations in an instance of best-response dynamics are separated, and are single-minded declarations.  Also, we can assume that all single-minded declarations are for sets that are either $M$ or of size at most $\sqrt{m}$, since these are the only sets allocated by $\calm_{CA}$.

\begin{prop}
\label{prop.sep.mech.CA}
If $\decls$ is separated, then for all $j$ such that $S_j \neq \emptyset$, there exists some $S_i$ such that $\decli(S_i) > \decl_j(S_j)$, $\mech_{CA}$ allocates $S_i$ to agent $i$.
\end{prop}
\begin{proof}
If $\mech_{CA}$ allocates $M$ to some bidder than the result is trivial, so suppose not.  If $S_j = M$, then the result again follows trivially.  If $|S_j| \leq \sqrt{m}$, then, from the perspective of agent $j$, $\mech_{CA}$ behaves precisely as $\mech_{\sqrt{m}CA}$, so the result follows from Proposition \ref{prop.sep.mech.sCA}.
\end{proof}


Fix some separated declaration profile $\decls$.  We now wish to characterize the utility-maximizing declarations of an agent $i$.  We begin with the simple case of declarations for set $M$.

\begin{lem}
\label{lem.CAbig}
If $SW_{\mech_{CA}}(\emptyset,\declsmi) < \frac{1}{4}\typei(M)$, then the utility-maximizing bid for agent $i$, $\decli$, sets $\decli(S_i) \geq \frac{1}{2}\typei(M)$ for some $S_i$.
\end{lem}
\begin{proof}
We claim that $\crit_i^{\mech_{CA}}(M,\declsmi) \leq \frac{1}{2}\typei(M)$.
To see this, observe that since $SW_{\mech_{CA}}(\emptyset,\declsmi) < \frac{1}{4}\typei(M)$ and $\decls$ is separated, it must be that no agent (other than $i$) bid for $M$ at a value greater than $\frac{1}{4}\typei(M)$, and that the sum of all bids for $M$ (by agents other than $i$) must be at most twice $\frac{1}{4}\typei(M)$.  The threshold amount for allocating $M$ to a single bidder can also be at most $SW_{\mech_{CA}}(\emptyset,\declsmi)$. So $\criti^{\mech_{CA}}(M,\declsmi) \leq 2\left(\frac{1}{4}\typei(M)\right)$ as claimed.  

This implies that agent $i$ can make a utility of $\frac{1}{2}\typei(M)$ by bidding on set $M$. Thus his utility-maximizing declaration must be a bid for a set $S_i$ with $\typei(S_i) \geq \frac{1}{2}\typei(M)$.  The result then follows from Lemma \ref{lem.undom1}.
\end{proof}

Next consider allocations of small sets.  Let $A_1, \dotsc, A_n$ be the optimal allocation of sets of size at most $\sqrt{m}$.  Recall the definition of $R_i$ from our analysis of $\mech_{sCA}$ in the previous section.  

\begin{lem}
\label{lem.CAsmall}
If $\sum_{j \in R_i(\decls,A_i)}\decl_j(S_j) < \frac{1}{2}\typei(A_i)$, then the utility-maximizing bid for agent $i$, $\decli$, sets $\decli(S_i) \geq \frac{1}{2}\typei(A_i)$ for some $S_i$.
\end{lem}
\begin{proof}[Proof (of Lemma \ref{lem.CAsmall})]

Recall the definitions of $R_i$ and $Q_i$ from our analysis of $\mech_{sCA}$ in the previous section.  For ease of notation we will write $SW_{\sqrt{m}}$ for $SW_{\mech_{\sqrt{m}CA}}$ in this proof.
We begin with a technical lemma that bounds the effect of a single agent's bid on the social welfare obtained by $\mech_{CA}$.

\begin{claim}
\label{claim.CAtech}
Suppose $\decls$ is separated and $\mech_{\sqrt{m}CA}(\decls)$ allocates $S_i$ to agent $i$.  Then $\decli(S_i) - \sum_{j \in R_i(\decls,S_i)}\decl_j(S_j) \leq SW_{\sqrt{m}}(\decli,\declsmi) - SW_{\sqrt{m}}(\emptyset,\declsmi) \leq \decli(S_i)$.
\end{claim}
\begin{proof}

Let $T_1, \dotsc, T_n$ be the sets allocated in $\mech_{\sqrt{m}CA}(\emptyset,\declsmi)$, and let $U_1, \dotsc, U_n$ be the sets allocated in $\mech_{\sqrt{m}CA}(\decli,\declsmi)$.  Thus $SW_{\mech_{\sqrt{m}CA}}(\decli,\declsmi) = \sum_j \decl_j(U_j)$ and $SW_{\mech_{\sqrt{m}CA}}(\emptyset,\declsmi) = \sum_j \decl_j(T_j)$.  

Since $S_i$ is allocated to agent $i$, it must be that $\decli(S_i) \geq \sum_{j \in R_i(\decls,S_i)}\decl_j(S_j)$.  Let us first consider the case that $\decli(S_i) = \sum_{j \in R_i(\decls,S_i)}\decl_j(S_j)$.  Note that, for any $j$, if $T_j \neq \emptyset$ and $U_j = \emptyset$, then it must be that there exists some agent $k$ such that $j \in Q_k(\decls,S_k)$ and $U_k \neq \emptyset$ (by Proposition \ref{prop.sep.mech.CA}).  Thus, summing over all $j$, it must be that 
\begin{equation}
\begin{split}
\label{eq.1}
\sum_j \decl_j(T_j) & \leq \sum_k \sum_{j \in Q_k(\decls,S_k)} \decl_j(T_j) \\
& \leq \sum_k \decl_k(U_k).
\end{split}
\end{equation}
Since $U_i = S_i$, if we were to increase the value of $\decli(S_i)$, this increases $\sum_k \decl_k(U_k)$ by the same amount.  Thus, if we write $x = \decli(S_i) - \sum_{j \in R_i(\decls,S_i)}\decl_j(S_j)$, \eqref{eq.1} becomes $\sum_k \decl_k(U_k) - x \geq \sum_j \decl_j(T_j)$.  Rearranging gives $\sum_k \decl_k(U_k) - \sum_j \decl_j(T_j) \geq x = \decli(S_i) - \sum_{j \in R_i(\decls,S_i)}\decl_j(S_j)$, which gives one half of our desired inequality.

For the other inequality, suppose for contradiction that $SW_{\sqrt{m}}(\decli,\declsmi) - SW_{\sqrt{m}}(\emptyset,\declsmi) > \decli(S_i)$.  Define $\decls'$ by setting $\decl_j' = \decl_j$ when $S_j \cap S_i = \emptyset$, and $\decl_j' = \emptyset$ when $S_j \cap S_i \neq \emptyset$.  That is, $\decls'$ is $\decls$ with all bids intersecting $S_i$ removed.  
As we showed in \eqref{eq.1}, replacing any winning bid with $\emptyset$ can only decrease the social welfare obtained by $\mech_{\sqrt{m}CA}$.  This implies (by removing bids one at a time) that $SW_{\sqrt{m}}(\emptyset,\declsmi) \geq SW_{\sqrt{m}}(\emptyset,\declsmi')$.

Note that allocation $\mech_{\sqrt{m}CA}(\decli,\declsmi)$ is precisely $\{S_i\} \cup \mech_{\sqrt{m}CA}(\emptyset,\declsmi')$.  We conclude that $SW_{\sqrt{m}}(\decli,\declsmi) - SW_{\sqrt{m}}(\emptyset,\declsmi) = \decli(S_i) - (SW_{\sqrt{m}}(\emptyset,\declsmi) - SW_{\sqrt{m}}(\emptyset,\declsmi')) \leq \decli(S_i)$, as required.  
This completes the proof of Claim \ref{claim.CAtech}.
\end{proof}

We now proceed with the proof of Lemma \ref{lem.CAsmall}.
Suppose for contradiction that $\decli$ is a bid for $S_i$ with $\decli(S_i) < \typei(A_i)/2$.  Then it must be that $S_i \neq M$, since $\typei(M) \geq \typei(A_i)$.  Let $\decli'$ be the single-minded bid for $A_i$ at value $\typei(A_i)$.  

We now consider cases depending on whether not $\calm_{CA}(\decli,\declsmi)$ and/or $\calm_{CA}(\decli',\declsmi)$ allocate $M$ to a single bidder.  If $\calm_{CA}(\decli,\declsmi)$ and $\calm_{CA}(\decli',\declsmi)$ either both allocate $M$ to one agent, or neither do, then $\utili(\decli',\declsmi) > \utili(\decli,\declsmi)$, a contradiction (where the case when they both allocate $M$ relies on the $\eps$-possibility that large bids are ignored).

Suppose $\calm_{CA}(\decli',\declsmi)$ allocates $M$ to one agent, but $\calm_{CA}(\decli,\declsmi)$ does not.  Since neither $A_i$ nor $S_i$ is $M$, it must be that $SW_{\calm_{\sqrt{m}CA}}(\decli',\declsmi) < SW_{\calm_{\sqrt{m}CA}}(\decli,\declsmi)$.  However, by Claim \ref{claim.CAtech}, $SW_{\calm_{\sqrt{m}CA}}(\decli',\declsmi) - SW_{\calm_{\sqrt{m}CA}}(\emptyset,\declsmi) \geq \decli'(A_i)-\sum_{j \in R_i(\decls,A_i)}\decl_j(S_j) \geq \typei(A_i)/2$, and  $SW_{\calm_{\sqrt{m}}}(\decli,\declsmi) - SW_{\calm_{\sqrt{m}}}(\emptyset,\declsmi) \leq \decli(S_i) < \typei(A_i)/2$.  Thus $SW_{\calm_{\sqrt{m}CA}}(\decli',\declsmi) \geq SW_{\calm_{\sqrt{m}CA}}(\decli,\declsmi)$, a contradiction.

Suppose instead that $\calm_{CA}(\decli,\declsmi)$ allocates $M$ to one agent, but $\calm_{CA}(\decli',\declsmi)$ does not.  Then it must be that agent $i$ obtains utility at least $\typei(A_i)/2$ when bidding for $\decli'$, and (since $S_i \neq M$) he obtains utility at most $\eps\typei(S_i) \leq \eps\typei(S_i)/2$ when bidding for $\decli$.  This contradicts the assumed maximality of $\decli$.  We have thus reached a contradiction in all cases.
\end{proof}

We are now ready to complete the proof of Theorem \ref{thm.best-response.2}, in a manner similar to Theorem \ref{thm.best-response.1}.

\begin{proof}[Proof (of Theorem \ref{thm.best-response.2})]

Since each $\decls^t$ is separated, it is clear that $SW_{\mech_{CA}}(\decls^t) \geq \decl_i^t(S_i)$ for all $t$.  Applying Lemma \ref{lem.numbad} to Lemma \ref{lem.CAbig} therefore implies that at for least $(\frac{1}{2}-\eps)T$ steps of $D$, $SW_{\mech_{CA}}(\decls^t) \geq \frac{1}{4}\typei(M)$.  We conclude that, with probability at least $1-ne^{-T\eps^2/32n}$, the average welfare obtained by $\mech_{CA}$ is within a constant factor of $\max_i \typei(M)$.  It must therefore also be an $O(\sqrt{m})$ approximation to the optimal welfare attainable through any allocation of sets of size at least $\sqrt{m}$ (as there can be at most $\sqrt{m}$ such sets in any valid allocation).

Now let $A_1, \dotsc, A_n$ be any allocation of sets of size at most $\sqrt{m}$, and for each $t$ let $G_t$ be the set of agents for which either $\decli^t(S_i^t) \geq \typei(A_i)/2$ or $\sum_{j \in R_i(\decls^t,S_i^t)}\decl_j(S_j^t) \geq \frac{1}{2}\typei(A_i)$.  Applying Lemma \ref{lem.numbad} to Lemma \ref{lem.CAbig}, we see that for each agent $i$ appears in $G_t$ for at least $(\frac{1}{2}-\eps)T$ time steps, with probability at least $1-ne^{-T\eps^2/32n}$.  The same argument as in the proof of Lemma \ref{lem.goodevents.welfare} then demonstrates that $\mech_{CA}$ obtains an average $O(\sqrt{m})$ approximation to $\sum_i \typei(A_i)$, over the steps of the best-response dynamics, with probability at least $1-ne^{-T\eps^2/32n}$.

Taking the union bound over the events described above, we conclude that $\mech_{CA}$ obtains an $O(\sqrt{m})$ approximation to any allocation of sets with size at most $\sqrt{m}$, and an $O(\sqrt{m})$ approximation to any allocation of sets of size at least $\sqrt{m}$, with probability at least $1-2ne^{-T\eps^2/n}$.  Since any allocation is a combination of sets of size at most $\sqrt{m}$ and sets of size at least $\sqrt{m}$, we conclude that $SW_{\mech_{CA}}(D) \geq \left(\frac{1}{O(\sqrt{m})}-\eps\right)SW_{opt}(\types)$ with probability at least $1-2ne^{-T\eps^2/n}$, as required.
\end{proof}





\bibliographystyle{plain}
\bibliography{auctions}

\end{document}